\documentclass[letterpaper]{article} 
\usepackage{aaai25}  
\usepackage{times}  
\usepackage{helvet}  
\usepackage{courier}  
\usepackage[hyphens]{url}  
\usepackage{graphicx} 
\usepackage{natbib}  
\usepackage{caption} 
\usepackage{algorithm}
\usepackage{algorithmic}

\usepackage{newfloat}
\usepackage{listings}
\usepackage{pifont}
\usepackage{ifthen}
\usepackage{etoolbox}
\newbool{EXPLAIN}    
\setbool{EXPLAIN}{true}  

\DeclareCaptionStyle{ruled}{labelfont=normalfont,labelsep=colon,strut=off} 
\floatstyle{ruled}
\newfloat{listing}{tb}{lst}{}
\floatname{listing}{Listing}
\usepackage{libertine}
\usepackage{microtype}

\usepackage[utf8]{inputenc}
\usepackage{url}            
\usepackage{booktabs}

\usepackage[dvipsnames]{xcolor}    

\usepackage{multirow}
\usepackage{trimclip}

\usepackage{amssymb,amsmath,amsfonts,amstext,amsthm, bm}
\usepackage{natbib}
\usepackage[inline]{enumitem}
\usepackage{booktabs, makecell}
\usepackage{cleveref}

\usepackage{tablefootnote}

\usepackage{extpfeil}

\renewcommand{\cite}[1]{\citep{#1}}

\DeclareMathOperator{\argmax}{\arg\max}

\usepackage{pgfplots}
\usepackage{pgfplotstable}

\usepackage{tikz}
\usetikzlibrary{shapes.misc, automata, positioning, arrows, arrows.meta}

\usepackage{breakcites}

\usepackage{thm-restate}

\newtheorem{lemma}{Lemma} 

\newtheorem{theorem}{Theorem}
\newtheorem{corollary}{Corollary}

\newtheorem{claim}{Claim}
\newtheorem{fact}{Fact}

\theoremstyle{definition}
\newtheorem{definition}{Definition}[section]
\newtheorem{example_}{Example}

\makeatletter
\DeclareRobustCommand{\shortto}{%
  \mathrel{\mathpalette\short@to\relax}%
}

\newcommand{\short@to}[2]{%
  \mkern2mu
  \clipbox{{.5\width} 0 0 0}{$\m@th#1\vphantom{+}{\shortrightarrow}$}%
  }
\makeatother

\newcommand{\indicatorconstant}{V}

\newcommand{\upar}{u_{frac}}
\newcommand{\uparp}[2]{\upar^{#1}(#2)}

\newcommand{\upresind}{u_{ind}}
\newcommand{\upresindp}[2]{\upresind^{#1}(#2)}

\newcommand{\upresplus}{u_{plus}}
\newcommand{\upresplusp}[2]{\upresplus^{#1}(#2)}

\newcommand{\upresmax}{u_{max}}
\newcommand{\upresmaxp}[2]{\upresmax^{#1}(#2)}

\title{Voter Priming Campaigns: Strategies, Equilibria, and Algorithms}
\author {
    Jonathan Shaki \textsuperscript{\rm 1},
    Yonatan Aumann \textsuperscript{\rm 1},
    Sarit Kraus \textsuperscript{\rm 1}
}
\affiliations {
    \textsuperscript{\rm 1}Bar-Ilan University\\
    jonath.shaki@gmail.com, yaumann@gmail.com, sarit@cs.biu.ac.il
}

\begin{document}
\maketitle
\begin{abstract}

    Issue \emph{salience} is a major determinant in voters' decisions.  Candidates and political parties campaign to shift salience to their advantage - a process termed \emph{priming}. We study the dynamics, strategies and equilibria of campaign spending for voter priming in multi-issue multi-party settings. We consider both parliamentary elections, where parties aim to maximize their share of votes, and various settings for presidential elections, where the winner takes all. For parliamentary elections, we show that pure equilibrium spending always exists and can be computed in time linear in the number of voters. For two parties and all settings, a spending equilibrium exists such that each party invests only in a single issue, and an equilibrium can be computed in time that is polynomial in the number of issues and linear in the number of voters. We also show that in most presidential settings no equilibrium exists. Additional properties of optimal campaign strategies are also studied.
\end{abstract}

\section{Introduction}
Political parties and candidates invest substantial resources in campaigns aimed at winning over voters. Interestingly, research reveals that these campaigns often do not directly alter voters' views or attitudes toward candidates and issues. Rather, campaigns tend to influence the relative \emph{salience} of various topics, thereby shaping the importance voters assign to these issues~\cite{IssueSalienceandPoliticalDecisions,macdonald2023core,Druckman2004}. 
Consequently, an essential element in political campaigns involves strategically guiding the salience of issues to the candidates' advantage, a process known in political science as \emph{priming}~\cite{Bartels2006}. 
Political campaigns are costly, and candidates have limited budgets.  Thus, from a strategic perspective, the question is how to best invest the available campaign budget in order to obtain the most beneficial priming for a given candidate, and whether an equilibrium investment profile exists. From a computational perspective, the question is how to compute the optimal and equilibrium investments, if they exist.  These are the topics of this paper. 

The importance of issue salience in voting has recently gained attention in the context of election control and manipulation~\cite{Lu2019,estornell2020election}. In these works, priming is viewed in the context of election manipulation and hence assumed to be limited to a single malicious actor.  We view priming in the context of campaigning and study equilibria thereof. 

\subsection{Summary of Results}
We consider two core settings, which we term \emph{parliamentary} and \emph{presidential}.  In the parliamentary setting, parties aim to maximize their share of votes.  In the presidential case, the primary goal is to be ranked first.  For the presidential setting, we consider several variants, which relate to a possible secondary goal of increasing the share of votes.\footnote{without such a secondary goal, unnatural behavior can emerge wherein candidates invest against their own good.} The different variants, formally defined in detail in Section \ref{sec:model}, are: (i) $\upresind$: the sole goal is to be ranked first, with no secondary goal, (ii) $\upresplus$: a weighted average the goal of being ranked first and increasing the share of votes, with the weight highly skewed towards the former, (iii) 
$\upresmax$: the primary goal is to be ranked first, only if this cannot be achieved then increasing the share is considered as a goal, with significantly lesser utility.  Throughout, we consider plurality ballot form, wherein each voter casts one vote for its most favorable candidate/party \cite{meir2010convergence}. 

For each setting and variant, we consider the question of whether a Nash equilibrium necessarily exists and whether, knowing the strategies of the other players, a best-response necessarily exists. Our main results are summarised in Table \ref{tab:summary}. For the parliamentary setting, both a Nash equilibrium (in pure strategies) and best-response always exist.  When there are only two candidates, it is further the case that in equilibrium each candidate invests its entire budget in priming of a single issue.  For the presidential setting, for both variants that include a secondary goal, a Nash equilibrium need not exist in the general case, and not even a best-response.
For the variant where the sole goal is winning, with no other considerations, a best response does exists, and we do not know if a Nash equilibrium necessarily exists.  We note, however, that the Nash equilibria in this variant may be unnatural, in the sense that they require candidates to invest not for their own good, but rather for the benefit of others. In all presidential variants, if there are only two candidates, an equilibrium exists wherein each candidate invests its entire budget in priming of a single issue. 

Whenever a Nash equilibrium is guaranteed to exist, we provide an algorithm that is linear in the number of voters to compute it. For the two candidates case, the algorithms are also polynomial in the number of issues and candidates. For the general parliamentarian case, the algorithm may be exponential in the number of issues and candidates. 

Due to limited space, the full proofs appear in appendix. In the main body we mostly provide limited proofs and outlines.

\begin{table*}[t]
\centering
\renewcommand{\arraystretch}{3}

\begin{tabular}{@{\extracolsep{6pt}}lc|ccc}

\toprule
    \makecell{setting}
    & \makecell{variant}
    & \makecell{best response}
    & \makecell{Nash equilibrium}
    & \makecell{Nash equilibrium \\ 2 candidates}
    \\
\midrule
\makecell{parliamentary} 
& \makecell{$\upar$} 
& \makecell{\checkmark exists \\ (\Cref{frac:best_response_exists})} 
& \makecell{\checkmark exists\\ (in pure strategies) \\ (\Cref{frac:pure_equilibrium_exists})} 
& \makecell{\checkmark exists  (with pure strategies, \\ single issue investments) \\ (\Cref{frac:2_candidates_focus_pure_equilbirium})} 
\\

\hline

\multirow{3}{*}{presidential} 

&\makecell{$\upresplus$}  
&  \makecell{\ding{55} need not \\ exist \\ (\Cref{plus:no_best_response})} 
& \makecell{\ding{55} need not \\ exist \\ (\Cref{plus:no_nash_equilibrium})} 
& \makecell{\checkmark exists (with pure strategies, \\ single issue investments) \\ (\Cref{plus:2_candidates_frac_equilbirium})} 
\\

\cline{2-5}

&\makecell{$\upresmax$} 
& \makecell{\ding{55} need not \\ exist \\ (\Cref{max:best_response_doesnt_exist})} 
& \makecell{\ding{55} need not \\ exist \\ (\Cref{max:nash_equilibrium_doesnt_exist})} 
& \makecell{\checkmark exists (with pure strategies, \\single issue investments) \\ (\Cref{max:two_candidates})} 
\\

\cline{2-5}

& \makecell{$\upresind$} 
& \makecell{\checkmark exists \\ (\Cref{ind:best_response_exists})} 
& $\bm{?}$ open 
& \makecell{\checkmark exists (with pure strategies,\\ single issue investments) \\ (\Cref{ind:2_candidates_dominant_strategies})} 
\\

\bottomrule
\end{tabular}
\caption{{\bf{Summary of results}.} Whenever Nash equilibrium is guaranteed to exist, it can be computed in time linear in the number of voters.  For two candidates, the time is also polynomial in the the number of issues. 
}
\label{tab:summary}
\end{table*}

\subsection{Related Work}
Campaign management and bribery have been researched extensively in the literature (see~\cite{schlotter2011campaign,faliszewski2015complexity,knop2020voting,zagoury2021targeted}. And for extensive review see~\cite{simpser2013governments, islam2010methods}).
Both these terms refer to a setting where individual voters are influenced to alter their vote, and the main question is how to choose these voters. We consider situations where changing voters' decision-making is only possible by general campaign messages, which alter the importance of different issues.

Ad delivery algorithms also aim to change voters' actions. 
They promise to lower the cost of advertising and increase the efficiency of campaigns
through detailed targeting, where advertisers can specify the
users they would like to reach using attributes \cite{danaher2023optimal, kreiss2011yes}. However, it was shown (e.g., \cite{ali2021ad,cotter2021reach}) that
targeting may not work as intended, and may shape the political ad delivery in ways that may not be beneficial to the political campaigns and to societal discourse. 

Campaign communications that direct voter attention to the considerations that
campaigns emphasize are termed ‘priming’. This phenomenon has been studied repeatedly over the years in many contexts (from psychology ~\cite{kuzyakov2000review, molden2014understanding} to natural language processing ~\cite{shaki2023cognitive, zhou2024context}). Regarding elections, Matthews \cite{matthews2019issue} showed clear evidence of issue priming throughout six US national elections. 
Macdonald \cite{macdonald2023core} showed that citizens’ core values are less affected by campaign-related priming effects, supporting our approach that the effect of investing in priming depends on the issue.

In \cite{hillygus2008persuadable}[Chapter. 1] the importance of wedge issues is discussed, showing the possibility to attract attention to topics that are less important and even those that are not the strongest of the candidates. 

There are several works that model the utility of voters, which take issues' importance into consideration.
For example, in \cite{kollman2018computational}'s
model, termed spatial elections, each voter attaches both a weight and an aspired value
to each of the issues, and every candidate has a position on every issue. Each voter votes for the candidate with minimal weighted Euclidean distance between its positions and the voter's ideal positions. They consider the problem of finding the optimal positions for candidates in order to maximize their votes.

\cite{Lu2019} consider issue salience priming/manipulation in the spatial preference framework.  assuming a single primer/manipulator. They consider the case of binary salience, and a manipulator that can determine which issues the voters care about and which not. Voters vote deterministically to the candidate who is closest to them (on the salient issues). \cite{Lu2019} view priming (which they term \emph{control by manipulation}) as a type of attack by a malicious entity, and hence only consider priming by a single entity.  The paper derives strongly negative results, and concludes that computing effective manipulation is hard even for two candidates, or for a single voter.

\cite{estornell2020election} extend \cite{Lu2019} and consider issue priming in the setting of weighted spatial preferences (rather than binary), and a manipulator that is constrained by a budget (for shifting the weights).  They consider both the deterministic setting - where the voter necessarily votes for the closest candidate - 
and stochastic settings  - where the probability of voting for a candidate is monotone in the weighted distance, 
and both the parliamentary and the presidential elections  (which they term \emph{MaxSupport} and \emph{Majority Vote}).   
Following \cite{Lu2019}, \cite{estornell2020election} consider only a single manipulator. They show that in the most general setting, the control problem is NP-hard, but that the stochastic setting - which is the more similar to ours - is tractable if the probability of voting for a candidate is linear in their weighted distance from the voter (there are also versions of the deterministic setting that are tractable). 

Our work differs from \cite{estornell2020election} and \cite{Lu2019} in several ways. Our voter preference model strictly subsumes the spatial model.  On the other hand, we only consider the stochastic case and assume that the probabilities are linear in the relative weighted distances (which corresponds to the stochastic setting for which the problem is tractable).  More importantly, we consider the \emph{strategic setting} where priming is performed by multiple candidates in the race and seek an equilibrium in such a setting - a topic not considered in \cite{estornell2020election} and \cite{Lu2019} at all.

\cite{denter2020campaign} considers the problem of campaign investment management in a setting where campaigns have two simultaneous effects: (i) Persuasion: increasing the quality of the policy in the issue as perceived by the voters through policy
advertising and (ii) Priming: making the issue more salient, thereby increasing the issue's perceived importance (as in our model). The paper considers only the case of two candidates and two issues, and does not provide any algorithms. We note that the exact model of how campaigns alter the salience is somewhat different between our model and that of \cite{denter2020campaign}, which leads to different results concerning the existence of an equilibrium.

\section{Model}
\label{sec:model}
We consider the following stylized model. There is a set $V$ of \emph{voters}, a set $C$ of \emph{candidates/parties}, and a set $I=\{ 1, \ldots , n\}$ of \emph{issues}. For each candidate $c$, issue $i$, and voter $v$, there is a \emph{quality score} $q^v_i(c)$, reflecting voter $v$'s perception of candidate $c$'s competence in handling issue $i$.  We assume that all quality scores are non-negative and $\sum_{c\in C}q^v_i(c) \leq 1$.  So, the scores reflect \emph{relative competence}, and can be viewed as the probability of $v$ to vote $c$ if issue $i$ were the only issue at stake; and $1-\sum_{c\in C}q^v_i(c)$ is the probability that $v$ does not vote. 

The salience of issue $i$ for voter $v$ represents how important this issue is for $v$. Before the campaign, this salience is $s_i^v(0)$. Since we only care about the relative salience (which is what governs the voting decision), we assume that $\sum_{i\in I}s_i^v(0)=1$ for every $v \in V$, and that $s^v_i(0) \ge 0$ for every $v \in V, i \in I$.

Each candidate $c$ has a total campaign budget $W^c \ge 0$, which it can distribute among the issues in order to increase their salience, with $w^c_i \ge 0$ denoting the amount invested in promoting issue $i$.  So, $\sum_{i\in I}w^c_i\leq W^c$ (where inequality occurs when some of the budget is not used). The total budget of all candidates is $W^* = \sum_{c' \in C} W^{c'}$. The total investment, over all candidates in issue $i$ is denoted $w_i$.  The vectors $\bm{w}^c=(w_1^c,\ldots,w_n^c), \bm{w}=\sum_{c \in C} \bm{w}^c$ are $c$'s investment profile and the total investment profile, respectively. In addition, $\bm{w}^{-c}=\sum_{c' \ne c} \bm{w}^{c'}$ is the total investment vector of all candidates except $c$.

Given a total investment of $w_i$ in issue $i$, the salience of this issue for voter $v$ is denoted $s_i^v(w_i)$. Investment in an issue linearly increases its salience: 
\begin{align}
    s_i^v(w_i)=\rho_iw_i+s_i^v(0) 
\end{align}
Here, $\rho_i \ge 0$ reflects the easiness, or difficulty, of increasing the salience of the issue; a large $\rho_i$ means that it is relatively easy to increase the salience, while a small $\rho_i$ - the opposite. The rationale is that some issues, say those that seem totally unimportant to the electorate, may require more investment to increase their salience than others.

Given investment $\bm{w}$, quality scores $q^v_i(c)$, and the salience scores $s^v_i(w_i)$, the relative salience is:

\begin{align*}
    s_i^v(\bm{w}) = \frac{s_i^v(w_i)}{\sum_{j\in I}s_j^v(w_j)}
\end{align*}

The probability that $v$ casts its vote for candidate $c$ is the weighted sum of the candidates' quality scores over the different issues, weighted by their relative salience:
\begin{align*}
    p^v(c,\bm{w}) = \sum_{i\in I} q^v_i(c) \cdot s_i^v(\bm{w}) = \bm{q^v}(c) \cdot \bm{s^v}(\bm{w})
\end{align*}
The following claim provides that the $p^v(c,\bm{w})$'s indeed correspond to the necessary probability structure:
\begin{claim}
    For any possible $\bm{w}$, $p^v(c,\bm{w})\geq 0$, for all $v,c$. In addition, $\sum_{c \in C} p^v(c,\bm{w}) \le 1$ for every $v$.
\end{claim}

Note that the probabilities need not sum to 1.  The remaining probability is the probability that the voter chooses not to vote at all.  This could happen, for example, if the voter deems all candidates incompetent in the salient issues.  If we wish to avoid such behavior we can further require that $\sum_{c\in C}q_i^v(c)=1$ for all $i$ and $v$. 
Hence the total expected number of votes candidate $c$ gets is
\begin{align*}
    p(c, \bm{w}) = \sum_{v\in V} p^v(c,\bm{w})
\end{align*}
We consider two settings: parliamentary elections, where each candidate (/party) aims to maximize its fraction of votes, and presidential elections, where candidates seek to win the elections. It is assumed throughout that the number of voters is sufficiently large so that the actual number of votes any candidate gets is essentially the expectation. 

The expected fraction of voters voting for candidate $c$ is thus
\begin{align*}
    r(c, \bm{w}) = \frac{p(c, \bm{w})}{\sum_{c'\in C} p(c', \bm{w})}
\end{align*}
The victory indicator is:
\begin{align*}
    v(c, \bm{w}) = \begin{cases}
        \frac{1}{|\underset{c' \in C}{\argmax} \ p(c', \bm{w})|} & c \in \underset{c' \in C}{\argmax} \ p(c', \bm{w}) \\
        0 & \text{otherwise}
    \end{cases}
\end{align*}
\paragraph{Parliamentary Elections}
Here, the utility of $c$ is simply
\begin{align}
    \uparp{c}{\bm{w}} = r(c, \bm{w})
\label{eq:u_par}
\end{align}



\paragraph{Presidential Elections} We consider several versions of presidential setting. The simplest is just taking the victory indicator:
\begin{align*}
    \upresindp{c}{\bm{w}} = v(c, \bm{w})
\end{align*}
This means, however, that candidates with no chance to win are agnostic to the number of votes they get.
This view overlooks the reality that even candidates with little chance of winning care deeply about their vote count. A strong voter base provides lasting political capital, influencing both the candidate's future prospects and their party's direction \cite{anagol2016runner}. 
In order to mitigate this, we propose two alternative utility functions. 
The first is just the sum of the indicator multiplied by a constant with the relative fraction of votes:
\begin{align*}
    \upresplusp{c}{\bm{w}} = V \cdot v(c, \bm{w}) + r(c, \bm{w})
\end{align*}
Where $\indicatorconstant \in R$ determines the importance of the victory indicator, versus that of the votes' fraction. With this utility function both the losing and the winning candidates care to maximize their votes' fraction.

And another one is the maximum between  the indicator multiplied by a constant and the number of relative votes:
\begin{align*}
    \upresmaxp{c}{\bm{w}} = \max(\indicatorconstant \cdot v(c, \bm{w}), r(c, \bm{w}))
\end{align*}
With this utility, only losing candidates care to maximize their fraction of votes, while a candidate that wins alone is agnostic.

It is assumed that $\indicatorconstant \ge |C|$, so that even a tie with with all candidates yields higher utility than the highest votes' fraction without winning.
\paragraph{The Priming Game.}
The different utility functions define a game, with the players being the candidates/parties, and the pure strategies being the investment profiles $\bm{w}^c$. We say that an investment $\bm{w}^c$ is \emph{split} if at least two issues get a positive investments:  $\exists i\neq j$ with $w_i^c \geq w_j^c>0$.  The investment is \emph{focused} if it is not split. We denote by $S_c$ the investment strategy of $c$ and $\bm{S}={(S_c)}_{c\in C}$.  Note that the $S_c$'s may be mixed. The following is trivial.  
\begin{fact}
    The games associated with $\upar, \upresind, \upresplus$ are constant sum games, but not that of $\upresmax$.
\end{fact}
\subsection{Simplifications}
\label{model:simplification}
The expression for $r, v$, and accordingly also for all the different utility functions, involves separate terms for each voter, each of which is itself rather complex.  With potentially millions of voters, this may seem like a problem, especially if further calculations need to be performed on these expressions (e.g. computing the equilibria). Fortunately, the following theorems provide that this complexity can be greatly simplified, as follows.

\begin{restatable}{lemma}{modelInvestingInZero}
    \label{model:investing_in_zero}
    When introducing an issue $0$ for which all candidates are ranked $0$ by all voters (i.e. $q^v_0(c) = 0$ for all $v \in V, c \in C$), the utilities of all candidates are identical whether they do not invest part of their budget or whether they invest it in issue $0$.
\end{restatable}

The above lemma assures us that we can assume all players always invest all their budget, some of it potentially in issue $0$. From now on we will always assume $0$ is included in $I$ and all candidates invest their entire budget.

Now, we can simplify the problem even further, by aggregating the voters' rating:

\begin{restatable}{theorem}{modelFinalSimplification}

It holds that,

\label{model:final_simplification}
\begin{align*}
    & r(c, \bm{w}) = \frac{\bm{Q^c} \cdot \bm{w}}{\bm{Q^*} \cdot \bm{w}}
    & \\
    & v(c, \bm{w}) = \begin{cases}
        \frac{1}{|\underset{c' \in C}{\argmax} \ \bm{Q^{c'}} \cdot \bm{w} |} & c \in \underset{c' \in C}{\argmax} \ \bm{Q^{c'}} \cdot \bm{w} \\
        0 & \text{otherwise}
    \end{cases}
\end{align*}
For

\begin{align*}
    & Q^c_i = \frac{\sum_{v \in V} \sum_{i\in I} q^v_i(c) \cdot s_i^v(0)}{W^*} + \sum_{v \in V} q^v_i(c) \cdot \rho_i
    & \\
    & Q^*_i = \sum_{c \in C} Q^c_i
\end{align*}

\end{restatable}

Since all the different utility functions depend only on $r(c, \bm{w}), v(c, \bm{w})$, we don't need to actually consider individual voters after calculating $\bm{Q^c} : \forall c \in C$, and the representation size is independent of the number of voters.

And we say that,

\begin{definition}
    The \emph{rank} of an issue $i$ for candidate $c$ is $Q^c_i$.
\end{definition}

\section{Parliamentary Elections} \label{sec:parliamentary}
\subsection{Best Response Strategies}
We start with studying the players' best response strategies. 

\begin{restatable}{theorem}{fracBestResponseExists}
    \label{frac:best_response_exists}
In the parliamentary setting, a best response (against any finite support strategy profile) always exists.
\end{restatable}
\ifbool{EXPLAIN}{The proof follows from the continuity of $u_{frac}$.}{}

\paragraph{Split investments.}  While players can split their investment between issues, the following proposition establishes that doing so is never a better response to pure strategies of the other candidates than focusing the entire investment on a single issue.
The following theorem is key to the existence and computation of a pure Nash equilibrium in the parliamentary setting.
\begin{theorem}
    \label{frac:focus_optimal_iff_split_is}
    For any candidate $c$, pure investment $w^{-c}$ of the other candidates, and response $w^c$ with strictly positive investment only in a set $J$ of issues:
    $w^c$ is a best response (for $c$) if and only if for any issue $j\in J$ a focused investment in $j$ (alone) is a best response. 
\end{theorem}

The proof of the theorem rests on two main lemmas.

\begin{restatable}{lemma}{fracBestMiddleResponse}
    \label{frac:best_middle_response}

    Given two feasible investments $\bm{x}, \bm{y}$ of $\bm{c}$, and an investment $\bm{z} = r \cdot \bm{x} + (1 - r) \cdot \bm{y} : 0 < r < 1$ between them, as responses against pure strategies $\bm{w^{-c}}$ of the other candidates, $\bm{z}$ is a best response if and only if $\bm{x}, \bm{y}$ are best responses.
\end{restatable}
This follows from the structure of $\upar$, and the fact that $
\frac{\partial \uparp{c}{r \cdot \bm{x} + (1 - r) \cdot \bm{y}} }{\partial r}
$ does not change sign for $0\leq r \leq 1$.  

\begin{restatable}{lemma}{feasibleInvestmentInVicinity}
    Let $\bm{z}$ be a feasible investment with a strictly positive investment in a set $J \subseteq I$ of issues, and let $\bm{x}$ be an investment with positive investments only in issues in $J$. 
    
    There exists $\epsilon > 0$ such that $\bm{y} = \bm{z} + \epsilon \cdot (\bm{x} - \bm{z})$ is a feasible investment.
\end{restatable}

We return to the proof of \Cref{frac:focus_optimal_iff_split_is}:

\begin{proof}
    Let $\bm{z}$ be a best response to $\bm{w^{-c}}$ with a strictly positive investment in a set $J \subseteq I$ of issues, and let $\bm{x}$ be a response with positive investments only in issues in $J$.

    By the above lemma, there exists $\epsilon > 0$ such that $\bm{y} = \bm{z} + \epsilon \cdot (\bm{x} - \bm{z})$ is a feasible investment.

    Since $\bm{z}$ is between $\bm{x}, \bm{y}$, by \Cref{frac:best_middle_response} it holds that $\bm{x}, \bm{y}$ are also best responses, and in particular $\bm{x}$ is.
    
    For the other direction, the proof is by induction. We assume that every response involves positive investment only in up to $n$ issues in $J$ is optimal, and show that every response with positive investment in up to $n+1$ issues in $J$ is optimal.

    Let $\bm{z}$ be an investment in $n+1$ issues in $J$, assumed without a loss of generality to be the first $n+1$ issues. Let $\bm{x} = (z_1 + z_{n+1}, z_2, z_3, ..., z_n, 0, 0, ..., 0)$ and $\bm{y} = (0, z_2, z_3, ..., z_n, z_1 + z_{n+1}, 0, 0, ..., 0)$. It can be seen that $\bm{z}$ is between $\bm{x}, \bm{y}$, and that $\bm{x}, \bm{y}$ each involves positive investments only in a set of $n$ issues, and must then be best responses. Hence, by \Cref{frac:best_middle_response} it holds that $\bm{z}$ is also a best response.

    The claim for $n=1$ holds by assumption.
\end{proof}

The theorem implies directly that there always exists a best-response which is focused. Accordingly, finding a best-response strategy is easy: compute the utility of all focused investments and choose the best.

The best response strategies can further be characterized using the following definition and lemmas.

\begin{restatable}{proposition}{utilityDependsOn}
    The utility of a candidate $c$ from investing in issue $i$ is characterized only by $Q^c_i$ and $Q^{-c}_i = \sum_{c'\neq c} Q^{c'}_i$.
\end{restatable}

\begin{restatable}{proposition}{superiorIssuesAreBetter}
\label{frac:superior_issues_are_better}
    For a candidate $c$ and issues $i, j$, if $Q^c_i \ge Q^c_j, Q^{-c}_i \le Q^{-c}_j$, there exists a best response with no investment in $j$.
    
    If one of the inequalities is strict, every response with a positive investment in $j$ is not a best response.
\end{restatable}

In other words, considering issues, candidates care only about their own ranking and the sum of the ranking of others, not about how this sum is split. With everything else equal, it's better to invest in issues for which your ranking is higher, or issues for which the sum of the others' rankings is lower.

\subsection{Nash Equilibrium}
\begin{theorem}
    \label{frac:pure_equilibrium_exists}
In the parliamentary setting,  there always exists a Nash equilibrium with pure strategies (but possibly split investments).
\end{theorem}

\begin{proof}
    The proof is by that of Theorem 1.2 in \cite{fudenberg1991game}, and similar to the standard proof of the Nash's theorem (e.g. Theorem 1.1 in the same work).

    The strategies space are indeed convex, as the only constrain for every candidate $c$ is that $\sum_{i \in I} w^c_i = W^C$, and it's easy to see that the utility of every candidate is continuous in the investments of all the candidates.

    Although the utility of every candidate is not quasi-concave, this requirement is used in the proof only to show that the optimal responses of every candidate are convex. Indeed, say that there are two best responses, one involves a set $J$ of issues with positive investments and the other a set $K$. Then, by \Cref{frac:focus_optimal_iff_split_is} every focus investment in any issue in $J \cup K$ is a best response. Then, by the same theorem any split investment in these issues is, hence the best responses are convex.
\end{proof}

\paragraph{Finding the Equilibrium.}
By \Cref{frac:pure_equilibrium_exists}, a pure equilibrium always exists. We explain how to find it.

In a pure equilibrium, every candidate $c$ strictly invests in a set $I^c \subseteq I$ of issues. By \Cref{frac:focus_optimal_iff_split_is}, $c$ is indifferent to any investment that involves only the issues in $I^c$.

The algorithm is as follows.

\paragraph{\bf Algorithm 1: computing Nash equilibrium.}

\label{frac:nash_equilibrium_algorithm}

Iterate over all possible (nonempty) sets of strictly invested issues $I^c \subseteq I : c \in C$. For every iteration, define the following set of linear equations and constraints. 

The variables are the investment of the different candidates, $w^c_i : c \in C, i \in I^c$. For ease of illustration, we let $B^c = \sum_{c' \ne c}\sum_{i' \in I^{c'}}w^{c'}_{i'} \cdot Q^c_{i'}$ and $B^* = \sum_{c' \ne c}\sum_{i' \in I^{c'}}w^{c'}_{i'} \cdot Q^*_{i'}$.

The constraints are:
\begin{enumerate}
    \item The investments are valid: candidates invest exactly their budget, and the investments are non-negative. 
\begin{align*}
    & \sum_{i \in I^c} w^c_i = 1, w^c_i \ge 0 & \forall c \in C, i \in I^c
\end{align*}
\item Candidates are indifferent between focused investments in issues they (strictly) invest in. Let an arbitrary $i^c \in I^c$ be the representative issue of $c$. For $c \in C, i^c \ne i \in I^c$:
\begin{align*}
    & W^c \cdot Q^*_i \cdot B^c + W^c \cdot Q^c_{i^c} \cdot B^* + (W^c)^2 \cdot Q^c_{i^c} \cdot Q^*_i
    & \\
    & = W^c \cdot Q^*_{i^c} \cdot B^c + W^c \cdot Q^c_i \cdot B^* + (W^c)^2 \cdot Q^c_i \cdot Q^*_{i^c}
\end{align*}
\item The candidates prefer their current investment rather than any other focus investment. For $c \in C$, a single representative $i \in I^c$, and every $j \not\in I^c$:
\begin{align*}
    & W^c \cdot Q^*_i \cdot B^c + W^c \cdot Q^c_j \cdot B^* + (W^c)^2 \cdot Q^c_j \cdot Q^*_i
    & \\
    & \le W^c \cdot Q^*_j \cdot B^c + W^c \cdot Q^c_i \cdot B^* + (W^c)^2 \cdot Q^c_i \cdot Q^*_j
\end{align*}
For every iteration, if the above set of equation has a solution, it is a Nash equilibrium and we can stop.

\end{enumerate}

The algorithm is exponential in the number of candidates and issues, as we iterate over all possible subsets of issues for every candidate. However, it is linear in the number of voters, as the simplified problem representation size is independent of the number of voters (\Cref{model:simplification}).

\subsection{Two Candidates}

\subsubsection{Best Response Strategies} 

When there are only two candidates, the structure of the best response strategies can be further refined.

\newcommand{\investin}[4]{\bm{w}(#1\shortto #2,#3\shortto #4)}
For candidates $c,c'$ and issues $i,i'$ denote $\investin{c}{i}{c'}{i'}$ the investment wherein $c$ investment all and only in $i$ and $c'$ all and only in $i'$.

\begin{restatable}{lemma}{investingInHigherIssue}
\label{frac:investing_in_higher_issue}
     For candidates $c,c'$ and issues $i, i', j, j'$, if
     \begin{align}
     & \label{frac:investing_in_higher_issue:prerequisite1}
     Q^c_i < Q^c_{i'}
     & \\
     & \label{frac:investing_in_higher_issue:prerequisite2}
     \upar^c(\investin{c}{i}{c'}{j})<
     \upar^c(\investin{c}{i'}{c'}{j}) & \text{and}, &\\
     & \label{frac:investing_in_higher_issue:prerequisite3}
     \upar^{c'}(\investin{c}{i'}{c'}{j})<
     \upar^{c'}(\investin{c}{i'}{c'}{j'})
     \end{align}
Then
\begin{align}
    \label{frac:investing_in_higher_issue:result}
     \upar^c(\investin{c}{i}{c'}{j'})&<
     \upar^c(\investin{c}{i'}{c'}{j'}) &
\end{align}
\end{restatable}

In other words, if one candidate prefers one issue over another, and the other candidate responds, the first candidate still prefers the first issue over the second, assuming its ranking on it is higher. This lemma will allow us to compute the Nash equilibrium, as we later show.

\subsubsection{Nash Equilibrium}

We will show that every 2-candidates game has a pure and focus equilibrium, such that every candidate is investing all his budget in a single issue. This can be shown by \Cref{frac:investing_in_higher_issue}: we let each candidate begin by investing in the issue on they are ranked the lowest. We then let each, in turns, respond by investing in the least higher-ranked issue that is a better response than their current investment (if exists). By the above lemma, this process will necessarily terminate with a Nash equilibrium.

Formally, the algorithm to find an equilibrium, given $C = \{c, c'\}$, is as follows:

\begin{algorithm}
\caption{Compute Nash Equilibrium for two candidates.}
\label{frac:2_candidates_equilbirium_algorithm}

\begin{enumerate}
    \item Let $i \leftarrow \underset{i' \in I}{\text{argmin }} Q^c_{i'}, j \leftarrow \underset{j' \in I}{\text{argmin }} Q^{c'}_{j'}$

    \item While true:
        \begin{enumerate}
            \item $I' = \{i' \in I : \upar^c(\investin{c}{i'}{c'}{j}) > \upar^c(\investin{c}{i}{c'}{j}) \land Q^c_{i'} > Q^c_i\}$

            \item $J' = \{j' \in I : \upar^c(\investin{c}{i}{c'}{j'}) > \upar^c(\investin{c}{i}{c'}{j}) \land Q^c_{j'} > Q^c_j\}$

            \item If $I' \ne \emptyset$, $i \leftarrow \underset{i' \in I'}{\text{argmin }} Q^c_{i'}$
            \item Else if $J' \ne \emptyset$, $j \leftarrow \underset{j' \in J'}{\text{argmin }} Q^c_{j'}$
            \item Else return $(i, j)$
        \end{enumerate}
\end{enumerate}
\end{algorithm}

\begin{restatable}{theorem}{twoCandidatesFocusPureEquilibrium}
    \label{frac:2_candidates_focus_pure_equilbirium}
    \label{frac:2_candidates_algorithm_complexity}
    
In the parliamentary setting, in the case of 2 candidates there always exists a Nash equilibrium with pure strategies and focused investments. \Cref{frac:2_candidates_equilbirium_algorithm} results in such equilibrium with time complexity of  $O(|I|^2 + |V|)$.
\end{restatable}




    


\section{Presidential Elections} \label{sec:presidential}
We consider the various settings of presidential elections.  

The following example shows that, in contrast to $\upar$'s \Cref{frac:focus_optimal_iff_split_is}, candidates can benefit from splitting their investment between more than one issue.  
Here, the only way for $c_1$ to win (partially) is to invest half in $1$ and half in $2$, which is the only best response. This example holds for all three presidential utility functions - $\upresind, \upresplus$, and $\upresmax$.  
\begin{table}[h]
\centering
\begin{tabular}{@{\extracolsep{4pt}}rrrr}
\toprule
& $c_1$ & $c_2$ & $c_3$ \\ 
\midrule
$Q^c_1$ & $1$ & $2$ & $0$ \\
$Q^c_2$ & $1$ & $0$ & $2$ \\
$W^c$ & $1$ & $0$ & $0$ \\
\bottomrule
\end{tabular}
\end{table}


\subsection{Victory Indicator, $\upresind$}

\label{ind:section}

\subsubsection{Two Candidates}

The case of two candidates is especially important for the presidential case. The following is intuitive.
\begin{restatable}{theorem}{twoCandidatesDominantStrategies}
    \label{ind:2_candidates_dominant_strategies}
    For the $\upresind$ presidential utility, and the case of two candidates $c, c'$, a dominant strategy for $c$ is to invest all budget in the issue $i \in I$ that maximizes $Q^c_i - Q^{c'}_i$.
\end{restatable}
As the goal is to be the candidate with the most votes, and in a two-candidate game it is entirely aligned with candidates maximizing the difference between the votes of themselves and the other.
It is straightforward then that a Nash equilibrium always exists and how to find it.

\subsubsection{General Number of Candidates}

\begin{theorem}
    \label{ind:best_response_exists}
    For the $\upresind$ presidential utility, a best response (against any finite support strategy profile) always exists.
\end{theorem}

\begin{proof}
    If a best response does not exist, then an infinite series of strictly better responses exists. However, since the number of possible investment profiles of the other candidates is finite, and for every investment profile the number of possible utility values for the responding candidate is finite (losing and tie with another $k$ candidates $0\leq k<|C|$). Hence, the overall number of outcomes is finite, hence such an infinite series cannot exist. 
\end{proof}

A possible problem with the $\upresind$ utility function is that losing candidates are indifferent to the number of votes they get, which may result in a counter-intuitive equilibrium. Consider the following example:

\begin{table}[h]
\centering
\begin{tabular}{@{\extracolsep{4pt}}rrrr}
\toprule
& $c_1$ & $c_2$ & $c_3$ \\ 
\midrule
$Q^c_1$ & $1$ & $0$ & $0$ \\
$Q^c_2$ & $0$ & $1$ & $0$ \\
$Q^c_3$ & $0$ & $0$ & $1$ \\
$W^c$ & $1$ & $1$ & $1$ \\
\bottomrule
\end{tabular}
\end{table}

In the above game, the case where all candidates invest all their budget in issue $1$ is an equilibrium, as $c_1$ still wins regardless if one of the other candidates changes his action. However, this equilibrium looks unnatural since these investments are clearly against their own interest. Although it's an equilibrium, we don't expect real-world campaigns to result in such an equilibrium.

In order to make the resulting equilibria more natural, we will slightly change the utility function such that losing candidates also care about the number of votes they get.

\subsection{Indicator Plus Fraction, $\upresplus$}

When summing the victory indicator (multiplied by a constant larger than $|C|$) with the votes fraction, both winning and losing candidates still try to maximize their share of votes. Hence, an unnatural Nash equilibrium such as that discussed in \Cref{ind:section} is not possible.

\subsubsection{Two Candidates}
\label{plus:two_candidates}

For two candidates, the game is in fact similar to that of $\upar$, the parliamentary case, and not to that of $\upresind$, the other presidential game.

\begin{restatable}{theorem}{plusTwoCandidatesFracEquilibrium}
\label{plus:2_candidates_frac_equilbirium}
    Every pure, focused Nash equilibrium for two candidates under $\upar$  is also a Nash equilibrium under $\upresplus$.
\end{restatable}

Intuitively, $\upar, \upresind$ are similar in the case of two candidates because, when there is only one candidate to play against, maximizing the share of votes is entirely aligned with trying to be the candidate with the most votes.

Based on the above theorem we have:
\begin{corollary}
    \label{plus:2_candidates_equilbirium_algorithm}
    For the $\upresplus$ presidential utility, with two candidates, a pure Nash equilibrium with focus investments always exists, and can be found using $\Cref{frac:2_candidates_equilbirium_algorithm}$.
\end{corollary}

\subsubsection{General Number of Candidates}
For a general number of candidates, the situation differs significantly.

\begin{restatable}{theorem}{plusNoBestResponse}
    \label{plus:no_best_response}
    With the $\upresplus$ presidential utility, a best investment need not exist, even when only one candidate has a positive investment budget.
\end{restatable}

In the proof we show a game where, for the investing candidate, one issue is better in order to maximize his share of votes, and any minimal investment in the other issue is required to win. That way, the investing candidate wants to invest "almost all" in the first and the rest in the other issue. However, since it's always possible to invest even more in the first issue and even less in the second, a best investment does not exist.
It is straightforward now that,

\begin{corollary}
    \label{plus:no_nash_equilibrium}
    With the $\upresplus$ presidential utility, a Nash equilibrium need not exist.
\end{corollary}

Since with only one investing candidate a Nash equilibrium is characterizes a best response.

\subsection{Maximum of Indicator and Fraction, $\upresmax$}

When taking the maximum between the victory indicator (times $V$) and the votes fraction, only losing candidates will try to maximize their share of votes.

\subsubsection{Two Candidates}
\label{max:two_candidates}

Similarly to \Cref{plus:two_candidates}, every equilibrium under $\upar$ is an equilibrium under $\upresmax$, and the proof is done in a similar manner.


\subsubsection{General Number of Candidates}

The difference between $\upresmax$ and $\upresplus$ is that for the former winning candidates are agnostic to their fraction of votes. As a result, the case of \Cref{plus:no_best_response}, where the only investing candidate wants to invest as much as possible in one issue to maximize his fraction, and any positive amount in another to ensure his victory, does not happen here. Indeed,

\begin{restatable}{theorem}{maxPureBestResponseExists}
    \label{max:pure_best_response_exists}
    With the $\upresmax$ presidential utility, a best response to pure strategies always exists.
\end{restatable}

However,

\begin{restatable}{theorem}{maxBestResponseDoesntExist}
    \label{max:best_response_doesnt_exist}
    With the $\upresmax$ presidential utility, a best response need not exist.
\end{restatable}
As in a case where the responding candidate has a chance to win and to lose, depending on the investments of the other candidates, the game is similar to that of $\upresplus$, as he cares about both the fraction of votes and whether he wins.

Finally, we can see that,

\begin{restatable}{theorem}{maxNashEquilibriumDoesntExist}
    \label{max:nash_equilibrium_doesnt_exist}
    With the $\upresmax$ presidential utility, a Nash equilibrium need not exist.
\end{restatable}

For a case similar to the described above.

\section{Conclusions and Discussion}

This study investigates voter priming strategies and equilibrium in multi-party, multi-issue elections, providing insights into the computational and strategic aspects of campaign spending. Notably, in all cases where an equilibrium is guaranteed to exist, it takes the form of a pure equilibrium. This outcome simplifies theoretical analysis and computational implementation, offering a clear framework for predicting candidate behavior. Pure equilibria allow candidates to adopt intuitive strategies, eliminating the need for probabilistic approaches that are harder to interpret.

Despite this, the results reveal certain counterintuitive aspects of the model that merit closer examination. In particular, the analysis establishes that in the two-candidate setting equilibrium in single-issue investments always exists. This result, while mathematically sound, diverges from real-world observations where candidates typically distribute their campaign budget across multiple issues. The divergence may be due to the linearity assumption in the model, creating incentives for candidates to concentrate all resources on a single issue.
A more realistic model could postulate diminishing returns on salience investment, which would encourage a more balanced allocation of resources. We leave the study of such a model to future research.

Another promising avenue for future work is to consider scenarios where candidates invest resources not only in priming voter attention on specific issues but also in improving their perceived quality on the different issues. This addition would provide a richer framework that captures the interplay between issue salience and perceived candidate quality, reflecting a broader range of campaign strategies observed in practice.

\section{Acknowledgments}

This research has been partially supported by the Israel Science
Foundation under grant  2544/24.

\clearpage

\bibliography{main}

\clearpage



\section{Appendix}

\subsection{Proofs from \Cref{sec:model}}

\modelInvestingInZero*

We begin by proving the following lemma:

\begin{lemma}
\label{model:first_simplification}

It holds that,

\begin{align}
    & \label{eq:u-big} r(c, \bm{w}) = \frac{\overline{B}^c + \bm{\overline{Q}^c} \cdot \bm{w}}{{\overline{B}^*} + \bm{\overline{Q}^*} \cdot \bm{w}}
    & \\
    & v(c, \bm{w}) = \begin{cases}
        \frac{1}{|\underset{c' \in C}{\argmax} \ \overline{B}^c + \bm{\overline{Q}^c} \cdot \bm{w}|} & c \in \underset{c' \in C}{\argmax} \ \overline{B}^c + \bm{\overline{Q}^c} \cdot \bm{w} \\
        0 & \text{otherwise}
    \end{cases} & \nonumber
\end{align}
For all $c\in C$, where $\overline{B}^c = \sum_{v \in V} \sum_{i\in I} q^v_i(c) \cdot s_i^v(0)$, $\overline{Q}^c_i = \sum_{v \in V} q^v_i(c) \cdot \rho_i$ and ${\overline{B}^*} = \sum_{c' \in C} B^{c'}$, $\bm{\overline{Q}^*} = \sum_{c' \in C} \bm{\overline{Q}^{c'}}$.
\end{lemma}

\begin{proof}
By definition:

\begin{align*}
    & p(c, \bm{w}) = \sum_{v \in V} p^v(c, \bm{w}) = \sum_{v \in V} \sum_{i\in I} q^v_i(c) \cdot s_i^v(\bm{w})
    & \\
    & = \sum_{v \in V} \sum_{i\in I} q^v_i(c) \cdot \frac{s_i^v(w_i)}{\sum_{j\in I}s_j^v(w_j)} 
    & \\
    & = \sum_{v \in V} \sum_{i\in I} q^v_i(c) \cdot \frac{\rho_i w_i+s_i^v(0)}{\sum_{j\in I}(\rho_j w_j+s_j^v(0))}
    & \\
    & = \sum_{v \in V} \sum_{i\in I} q^v_i(c) \frac{\rho_i w_i+s_i^v(0)}{1 + \sum_{j\in I}\rho_j w_j}
    & \\
    & = \frac{\sum_{v \in V} \sum_{i\in I} q^v_i(c) (\rho_i w_i+s_i^v(0))}{1 + \bm{\rho} \cdot \bm{w}}
    & \\
    & = \frac{\sum_{v \in V} \sum_{i\in I} (q^v_i(c) \cdot \rho_i w_i + q^v_i(c) \cdot s_i^v(0))}{1 + \bm{\rho} \cdot \bm{w}}
    & \\
    & = \frac{\overline{B}^c + \sum_{i\in I} w_i \sum_{v \in V} q^v_i(c) \cdot \rho_i}{1 + \bm{\rho} \cdot \bm{w}}
    & \\
    & = \frac{\overline{B}^c + \sum_{i\in I} w_i \cdot \overline{Q}^c_i}{1 + \bm{\rho} \cdot \bm{w}}
    & \\
    & = \frac{\overline{B}^c + \bm{w} \cdot \bm{\overline{Q}^c}}{1 + \bm{\rho} \cdot \bm{w}}
\end{align*}

Hence:

\begin{align*}
    & r(c, \bm{w}) = \frac{p(c, \bm{w})}{\sum_{c'\in C} p(c, \bm{w})} = \frac{\frac{\overline{B}^c + \bm{w} \cdot \bm{\overline{Q}^c}}{1 + \bm{\rho} \cdot \bm{w}}}{\sum_{c'\in C} \frac{B_{c'} + \bm{w} \cdot \bm{\overline{Q}^{c'}}}{1 + \bm{\rho} \cdot \bm{w}}}
    & \\
    & = \frac{\frac{\overline{B}^c + \bm{w} \cdot \bm{\overline{Q}^c}}{1 + \bm{\rho} \cdot \bm{w}}}{\frac{\sum_{c'\in C} B_{c'} + \bm{w} \cdot \bm{\overline{Q}^{c'}}}{1 + \bm{\rho} \cdot \bm{w}}} = \frac{\overline{B}^c + \bm{\overline{Q}^c} \cdot \bm{w}}{\sum_{c'\in C} B_{c'} + \bm{\overline{Q}^{c'}} \cdot \bm{w}}
    & \\
    & = \frac{\overline{B}^c + \bm{\overline{Q}^c} \cdot \bm{w}}{{\overline{B}^*} + \bm{\overline{Q}^*} \cdot \bm{w}}
\end{align*}

For the second part of the lemma, we can see that 

\begin{align*}
    & \underset{c \in C}{\argmax} \ \overline{B}^c + \bm{\overline{Q}^c} = \underset{c \in C}{\argmax} \ \frac{\overline{B}^c + \bm{\overline{Q}^c} \cdot \bm{w}}{{\overline{B}^*} + \bm{\overline{Q}^*} \cdot \bm{w}}
    & \\
    & = \underset{c \in C}{\argmax} \ r(c, \bm{w}) = \underset{c \in C}{\argmax} \ \frac{p(c, \bm{w})}{\sum_{c'\in C} p(c', \bm{w})}
    & \\
    & =  \underset{c \in C}{\argmax} \ p(c, \bm{w})
\end{align*}

Hence

\begin{align*}
    & v(c, \bm{w}) = \begin{cases}
        \frac{1}{|\underset{c' \in C}{\argmax} \ p(c', \bm{w})|} & c \in \underset{c' \in C}{\argmax} \ p(c', \bm{w}) \\
        0 & \text{otherwise}
    \end{cases}
    & \\
    & = \begin{cases}
        \frac{1}{|\underset{c \in C}{\argmax} \ \overline{B}^c + \bm{\overline{Q}^c} \cdot \bm{w}|} & c \in \underset{c \in C}{\argmax} \ \overline{B}^c + \bm{\overline{Q}^c} \cdot \bm{w} \\
        0 & \text{otherwise}
    \end{cases}
\end{align*}

\end{proof}

And now we go back for the proof of \Cref{model:investing_in_zero}:

\begin{proof}
We can see that

\begin{align*}
    & \overline{Q}^c_0 = \sum_{v \in V} q^v_0(c) \cdot \rho_i = \sum_{v \in V} 0 \cdot \rho_i = 0
\end{align*}

For all $c \in C$. Hence:

\begin{align*}
    & \overline{Q}^*_0 = (\sum_{c' \in C} \bm{\overline{Q}^{c'}})_0 = \sum_{c' \in C} \overline{Q}^{c'}_0 = 0
\end{align*}

We know that

\begin{align*}
    & r(c, \bm{w}) = \frac{\overline{B}^c + \bm{\overline{Q}^c} \cdot \bm{w}}{{\overline{B}^*} + \bm{\overline{Q}^*} \cdot \bm{w}} = \frac{\overline{B}^c + \sum_{i \in I} \overline{Q}^c_i \cdot w_i}{{\overline{B}^*} + \sum_{i \in I} \overline{Q}^*_i \cdot w_i}
    & \\
    & = \frac{\overline{B}^c + \overline{Q}^c_0 \cdot w_0 + \sum_{0 \ne i \in I} \overline{Q}^c_i \cdot w_i}{{\overline{B}^*} + \overline{Q}^c_0 \cdot w_0 + \sum_{0 \ne i \in I} \overline{Q}^*_i \cdot w_i}
    & \\
    & = \frac{\overline{B}^c + \sum_{0 \ne i \in I} \overline{Q}^c_i \cdot w_i}{{\overline{B}^*} + \sum_{0 \ne i \in I} \overline{Q}^*_i \cdot w_i}
\end{align*}

Hence, all investments into $w_0$ does not change the utility of any of the players, exactly like not investing this amount in the first place.

\end{proof}

\modelFinalSimplification*

\begin{proof}

It holds that,

\begin{align*}
    & Q^c_i = \frac{\sum_{v \in V} \sum_{i\in I} q^v_i(c) \cdot s_i^v(0)}{W^*} + \sum_{v \in V} q^v_i(c) \cdot \rho_i 
    & \\
    & = \overline{Q}^c_i + \frac{\overline{B}^c_i}{W^*}
    & \\
    & Q^*_i = \sum_{c \in C} Q^c_i = \overline{Q}^*_i + \frac{\overline{B}^*_i}{W^*}
\end{align*}

And we can see that

\begin{align*}
    & r(c, \bm{w}) = \frac{\overline{B}^c + \bm{\overline{Q}^c} \cdot \bm{w}}{{\overline{B}^*} + \bm{\overline{Q}^*} \cdot \bm{w}} = \frac{W^* \cdot \frac{\overline{B}^c}{W^*} + \bm{\overline{Q}^c} \cdot \bm{w}}{W^* \cdot \frac{\overline{B}^*}{W^*} + \bm{\overline{Q}^*} \cdot \bm{w}}
    & \\
    & = \frac{(\frac{\overline{B}^c}{W^*}, \frac{\overline{B}^c}{W^*}, ..., \frac{\overline{B}^c}{W^*}) \cdot \bm{w} + \bm{\overline{Q}^c} \cdot \bm{w}}{(\frac{\overline{B}^*}{W^*}, \frac{\overline{B}^*}{W^*}, ..., \frac{\overline{B}^*}{W^*}) \cdot \bm{w} + \bm{\overline{Q}^*} \cdot \bm{w}}
    & \\
    & = \frac{((\frac{\overline{B}^c}{W^*}, \frac{\overline{B}^c}{W^*}, ..., \frac{\overline{B}^c}{W^*}) + \bm{\overline{Q}^c}) \cdot \bm{w}}{((\frac{\overline{B}^*}{W^*}, \frac{\overline{B}^*}{W^*}, ..., \frac{\overline{B}^*}{W^*}) + \bm{\overline{Q}^*}) \cdot \bm{w}}
    & \\
    & = \frac{\bm{Q^c} \cdot \bm{w}}{\bm{Q^*} \cdot \bm{w}}
\end{align*}

The second equation in the theorem is obtained in an identical manner to that of \Cref{model:first_simplification}.

\end{proof}

\subsection{Proofs from \Cref{sec:parliamentary}}

\fracBestResponseExists*

\begin{proof}
    The utility of a candidate is continuous in his investment, and the constraints for his investment form a closed set. 
\end{proof}

\fracBestMiddleResponse*

\begin{proof}
    Consider, for $0 \le t \le 1$:
        \begin{align*}
        & f(t) = \uparp{c}{\bm{z}} = \uparp{c}{t \cdot \bm{x} + (1 - t) \cdot \bm{y}}
        & \\
        & = \frac{\bm{Q^c} \cdot (\bm{w^{-c}} + t \cdot \bm{x} + (1 - t) \cdot \bm{y})}{\bm{Q^*} \cdot (\bm{w^{-c}} + t \cdot \bm{x} + (1 - t) \cdot \bm{y})}
        & \\
        & = \frac{a + t \cdot b + (1 - t) \cdot d}{A + t \cdot B + (1 - t) \cdot D}
    \end{align*}
    For $a = \bm{Q^c} \cdot \bm{w^{-c}}, b = \bm{Q^c} \cdot \bm{x}, d = \bm{Q^c} \cdot \bm{y}$ and $A = \bm{Q^*} \cdot \bm{w^{-c}}, B = \bm{Q^*} \cdot \bm{x}, D = \bm{Q^*} \cdot \bm{y}$.
    Now, it can be verified that
    \begin{align*}
        f'(t) = -\frac{(B+A) \cdot d - (D + A) \cdot b + (B - D) \cdot a}{((D - B) \cdot t - D - A)^2}
    \end{align*}
    Hence, the sign of $f'(t)$ is constant, and $f$ is strongly monotonic, or constant.

    If $\bm{x}, \bm{y}$ are best responses, they must result in identical utility, hence $f(0) = f(1)$, and the function is constant, hence $f(r) = f(0) = f(1)$ and $\bm{z}$ is also a best response.

    For the other direction, assume $\bm{z}$ is a best response, and at least one of $\bm{x}, \bm{y}$ is not, assuming without a loss of generality that $x$ is not. Then $f(r) \ge f(0), f(r) > f(1)$. However, this contradicts monotonicity, so $\bm{x}, \bm{y}$ both must be best responses.
\end{proof}

\feasibleInvestmentInVicinity*

\begin{proof}
    First, we find $\epsilon > 0$ for which the investment in every issue is not negative. 
    Since $z_i = 0$ means $x_i = 0$, for every $i$ such that $z_i = 0$ it holds that $z_i + \epsilon \cdot (x_i - z_i) \ge 0$. 
    Now, for every $z_j > 0$, it is straightforward from the Archimedean property that there exists $\delta_j > 0$ such that for every $\delta'_j \le \delta_j$, it holds that $z_j < \delta'_j \cdot (z_i - x_i) \ge 0$, hence $z_i + \epsilon \cdot (x_i - z_i) \ge 0$. We then take $\epsilon = \underset{j : z_j > 0}{\text{min }} \delta_j$.

    Second, we ensure that the total sum of the investment does not exceed the budget:
    \begin{align*}
        & \sum_{i \in I} y_i = \sum_{i \in I} (z + \epsilon \cdot (x - z))_i 
        & \\
        & = \sum_{i \in I} z_i + \epsilon \cdot (\sum_{i \in I} z_i - \sum_{i \in I} x_i) 
        & \\
        & = W^c + \epsilon \cdot (W^c - W^c) = W^c
    \end{align*}
We conclude that $y$ is a feasible investment.
\end{proof}

\utilityDependsOn*

\begin{proof}
Denoting $S$ the possible investments of the other players, such that every $\bm{w^{-c}}$ is their total investment with probability $p_{\bm{w^{-c}}}$, the expected utility of $c$ is, for investment $\bm{w^c}$:

\begin{align*}
    & \sum_{w^{-c} \in S} p_{\bm{w^{-c}}} \cdot \frac{\bm{Q^c} \cdot \bm{w^{-c}} + \bm{Q^c} \cdot \bm{w^c}}{\bm{Q^*} \cdot \bm{w^{-c}} + \bm{Q^*} \cdot \bm{w^c}}
    & \\
    & = \sum_{w^{-c} \in S} p_{\bm{w^{-c}}} \cdot \frac{B^c + Q^c_i \cdot w^c_i}{B^* + Q^*_i \cdot w^c_i}
    & \\
    & = \sum_{w^{-c} \in S} p_{\bm{w^{-c}}} \cdot \frac{B^c + Q^c_i \cdot w^c_i}{B^* + (\sum_{c' \in C} Q^{c'}_i) \cdot w^c_i}
    & \\
    & = \sum_{w^{-c} \in S} p_{\bm{w^{-c}}} \cdot \frac{B^c + Q^c_i \cdot w^c_i}{B^* + (Q^c_i + \sum_{c' \ne c} Q^{c'}_i) \cdot w^c_i}
\end{align*}

For $B^c = \bm{Q^c} \cdot \bm{w^{-c}} + \sum_{k \in I - \{i\}} Q^c_k \cdot w^c_k$ and $B^* = \bm{Q^*} \cdot \bm{w^{-c}} + \sum_{k \in I - \{i\}} Q^*_k \cdot w^c_k$.
\end{proof}

\superiorIssuesAreBetter*

\begin{proof}
Denoting $S$ the possible investments of the other players, such that every $\bm{w^{-c}}$ is their total investment with probability $p_{\bm{w^{-c}}}$, the expected utility of $c$ is, for investment $\bm{w^c}$:

\begin{align*}
    & \uparp{c}{\bm{w^c}|S} = \sum_{\bm{w^{-c}} \in S} p_{\bm{w^{-c}}} \cdot \frac{\bm{Q^c} \cdot \bm{w^{-c}} + \bm{Q^c} \cdot \bm{w^c}}{\bm{Q^*} \cdot \bm{w^{-c}} + \bm{Q^*} \cdot \bm{w^c}}
    & \\
    & = \sum_{\bm{w^{-c}} \in S} p_{\bm{w^{-c}}} \cdot \frac{b^{\bm{w^{-c}}, \bm{w^c}} + Q^c_i \cdot w^c_i + Q^c_j \cdot w^c_j}{B^{\bm{w^{-c}}, \bm{w^c}} + Q^*_i \cdot w^c_i + Q^*_j \cdot w^c_j}
\end{align*}

For $b^{\bm{w^{-c}}, \bm{w^c}} = \bm{Q^c} \cdot \bm{w^{-c}} + \sum_{k \in I - \{i, j\}} Q^c_k \cdot w^c_k$ and $B^{\bm{w^{-c}}, \bm{w^c}} = \bm{Q^*} \cdot \bm{w^{-c}} + \sum_{k \in I - \{i, j\}} Q^*_k \cdot w^c_k$.

Considering the expected utility from the two possible investments, $\bm{w^{c, 1}}, \bm{w^{c, 2}}$, the only difference is in the last two terms of the numerator and denominator of the fraction for every summand, as $b^{\bm{w^{-c}}, \bm{w^c}}, B^{\bm{w^{-c}}, \bm{w^c}}$ are independent of $w^c_i, w^c_j$. 

For every $\bm{w^{-c}} \in S$, we claim that its summand in $\uparp{c}{\bm{w^{c, 2}}|S}$ is greater than in $\uparp{c}{\bm{w^{c, 1}}|S}$. To see this:

\begin{align*}
    & \frac{b^{\bm{w^{-c}}, \bm{w^c}} + Q^c_i \cdot w^{c, 2} + Q^c_j \cdot w^{c, 2}}{B^{\bm{w^{-c}}, \bm{w^c}} + Q^*_i \cdot w^{c, 2}_i + Q^*_j \cdot w^{c, 2}_j}
    & \\
    & - \frac{b^{\bm{w^{-c}}, \bm{w^c}} + Q^c_i \cdot w^{c, 1}_i + Q^c_j \cdot w^{c, 1}_j}{B^{\bm{w^{-c}}, \bm{w^c}} + Q^*_i \cdot w^{c, 1}_i + Q^*_j \cdot w^{c, 1}_j}
    & \\ 
    & = (b^{\bm{w^{-c}}, \bm{w^c}} + Q^c_i \cdot w^{c, 2} + Q^c_j \cdot w^{c, 2}) & \\
    & \cdot (B^{\bm{w^{-c}}, \bm{w^c}} + Q^*_i \cdot w^{c, 1}_i + Q^*_j \cdot w^{c, 1}_j)
    & \\
    & - (b^{\bm{w^{-c}}, \bm{w^c}} + Q^c_i \cdot w^{c, 1}_i + Q^c_j \cdot w^{c, 1}_j)
    & \\
    & \cdot (B^{\bm{w^{-c}}, \bm{w^c}} + Q^*_i \cdot w^{c, 2}_i + Q^*_j \cdot w^{c, 2}_j)
    & \\ 
    & = b^{\bm{w^{-c}}, \bm{w^c}} \cdot (Q^*_i \cdot w^{c, 1}_i + Q^*_j \cdot w^{c, 1}_j) 
    & \\
    & + B^{\bm{w^{-c}}, \bm{w^c}} \cdot (Q^c_i \cdot w^{c, 2}_i + Q^c_j \cdot w^{c, 2}_j)
    & \\
    & - b^{\bm{w^{-c}}, \bm{w^c}} \cdot (Q^*_i \cdot w^{c, 2}_i + Q^*_j \cdot w^{c, 2}_j)
    & \\
    & - B^{\bm{w^{-c}}, \bm{w^c}} \cdot (Q^c_i \cdot w^{c, 1}_i + Q^c_j \cdot w^{c, 1}_j)
    & \\ 
    & = b^{\bm{w^{-c}}, \bm{w^c}}
    & \\
    & \cdot  (Q^*_i \cdot w^{c, 1}_i + Q^*_j \cdot w^{c, 1}_j - Q^*_i \cdot w^{c, 2}_i - Q^*_j \cdot w^{c, 2}_j)
    & \\
    & + B^{\bm{w^{-c}}, \bm{w^c}}
    & \\
    & \cdot (Q^c_i \cdot w^{c, 2}_i + Q^c_j \cdot w^{c, 2}_j - Q^c_i \cdot w^{c, 1}_i - Q^c_j \cdot w^{c, 1}_j)
    & \\ 
    & = b^{\bm{w^{-c}}, \bm{w^c}}
    & \\
    & \cdot  (Q^{-c}_i \cdot w^{c, 1}_i + Q^{-c}_j \cdot w^{c, 1}_j - Q^{-c}_i \cdot w^{c, 2}_i - Q^{-c}_j \cdot w^{c, 2}_j)
    & \\
    & + b^{\bm{w^{-c}}, \bm{w^c}}
    & \\
    & \cdot  (Q^{c}_i \cdot w^{c, 1}_i + Q^{c}_j \cdot w^{c, 1}_j - Q^{c}_i \cdot w^{c, 2}_i - Q^{c}_j \cdot w^{c, 2}_j)
    & \\
    & + B^{\bm{w^{-c}}, \bm{w^c}}
    & \\
    & \cdot (Q^c_i \cdot w^{c, 2}_i + Q^c_j \cdot w^{c, 2}_j - Q^c_i \cdot w^{c, 1}_i - Q^c_j \cdot w^{c, 1}_j)
    & \\ 
    & = b^{\bm{w^{-c}}, \bm{w^c}}
    & \\
    & \cdot  (Q^{-c}_i \cdot w^{c, 1}_i + Q^{-c}_j \cdot w^{c, 1}_j - Q^{-c}_i \cdot w^{c, 2}_i - Q^{-c}_j \cdot w^{c, 2}_j)
    & \\
    & + (B^{\bm{w^{-c}}, \bm{w^c}} - b^{\bm{w^{-c}}, \bm{w^c}})
    & \\
    & \cdot (Q^c_i \cdot w^{c, 2}_i + Q^c_j \cdot w^{c, 2}_j - Q^c_i \cdot w^{c, 1}_i - Q^c_j \cdot w^{c, 1}_j)
\end{align*}

It can be easily verified that $b^{\bm{w^{-c}}, \bm{w^c}} \le B^{\bm{w^{-c}}, \bm{w^c}}$, that $Q^{-c}_i \cdot w^{c, 1}_i + Q^{-c}_j \cdot w^{c, 1}_j \ge Q^{-c}_i \cdot w^{c, 2}_i - Q^{-c}_j \cdot w^{c, 2}_j$ and that $Q^c_i \cdot w^{c, 2}_i + Q^c_j \cdot w^{c, 2}_j \ge Q^c_i \cdot w^{c, 1}_i - Q^c_j \cdot w^{c, 1}_j$ (note that $w^{c, 1}_i +  w^{c, 1}_j = w^{c, 2}_i +  w^{c, 2}_j$). Hence, the above equation is positive.

(As $b^{\bm{w^{-c}}, \bm{w^{c, 1}}}, b^{\bm{w^{-c}}, \bm{w^{c, 2}}}$ are identical, we let $\bm{w^c}$ be either $\bm{w^{c, 1}}, \bm{w^{c, 2}}$)

\end{proof}

\begin{theorem}
    Algorithm 1 is exponential in the number of candidates and issues, and linear in the number of voters. The algorithm always returns a pure Nash equilibrium.
\end{theorem}

\begin{proof}
    In order to see that the algorithm work, we first examine the constraints. 
    
    It's easy to see that the validness constraints hold iff the investment is a feasible investment.

    Next, for the issue indifferent constraints, we can see that:

    \begin{align*}
        & \frac{B^c + W^c \cdot Q^c_{i^c}}{B^* + W^c \cdot Q^*_{i^c}} = \frac{B^c + W^c \cdot Q^c_i}{B^* + W^c \cdot Q^*_i} \Leftrightarrow
        & \\
        & (B^c + W^c \cdot Q^c_{i^c}) \cdot (B^* + W^c \cdot Q^*_i) 
        & \\
        & = (B^c + W^c \cdot Q^c_i) \cdot (B^* + W^c \cdot Q^*_{i^c}) \Leftrightarrow
        & \\
        & W^c \cdot Q^*_i \cdot B^c + W^c \cdot Q^c_{i^c} \cdot B^* + (W^c)^2 \cdot Q^c_{i^c} \cdot Q^*_i
        & \\
        & = W^c \cdot Q^*_{i^c} \cdot B^c + W^c \cdot Q^c_i \cdot B^* + (W^c)^2 \cdot Q^c_i \cdot Q^*_{i^c}
    \end{align*}

    In the last passage the only non-linear term, $B^c \cdot B^*$, canceled itself from both sides of the equation.
    This condition is necessary and enough by \Cref{frac:focus_optimal_iff_split_is} for an optimal response.

    Finally, for the alternative investments constraints:

    \begin{align*}
    & \frac{B^c + W^c \cdot Q^c_j}{B^* + W^c \cdot Q^*_j} \le \frac{B^c + W^c \cdot Q^c_i}{B^* + W^c \cdot Q^*_i} \Leftrightarrow
    & \\
    & (B^c + W^c \cdot Q^c_j) \cdot (B^* + W^c \cdot Q^*_i) 
    & \\
    & \le (B^c + W^c \cdot Q^c_i) \cdot (B^* + W^c \cdot Q^*_j) \Leftrightarrow
    & \\
    & W^c \cdot Q^*_i \cdot B^c + W^c \cdot Q^c_j \cdot B^* + (W^c)^2 \cdot Q^c_j \cdot Q^*_i
    & \\
    & \le W^c \cdot Q^*_j \cdot B^c + W^c \cdot Q^c_i \cdot B^* + (W^c)^2 \cdot Q^c_i \cdot Q^*_j
    \end{align*}
    Again in the last passage $B^c \cdot B^*$ cancels itself from both sides of the equation.
    
    Recall that by \Cref{frac:focus_optimal_iff_split_is} there exists an optimal focus investment.

    It can be seen that the number of constraints and variables is identical; $B^*, B^c : c \in C$ are directly associated with an equation, and for every candidate $c$ there is one validness equation and $|I^c|-1$ indifference equations, result in $|I^c|$ equations at total which is identical to the number of variables associated with $c$. No additional variables are presented.

    Hence, and since all introduced constraints hold if and only if the solution is an equilibrium, the algorithm will eventually find a (pure) Nash equilibrium that exists by \Cref{frac:pure_equilibrium_exists}.
\end{proof}

\investingInHigherIssue*

\begin{proof}
    Since candidates are investing in different issues ($c$ only in $i, i'$ and $c'$ only in $j, j'$), we can assume a budget of $1$ for every candidate, by 
    normalizing every issue accordingly (i.e. dividing the ranking of each player by the original budget).

    By negation, we assume \crefrange{frac:investing_in_higher_issue:prerequisite1}{frac:investing_in_higher_issue:prerequisite3} hold but not \cref{frac:investing_in_higher_issue:result}, hence:

    \begin{align*}
        & \frac{Q^c_i + Q^c_j}{Q^*_i + Q^*_j} = \upar^c(\investin{c}{i}{c'}{j})
        & \\
        & < \upar^c(\investin{c}{i'}{c'}{j}) = \frac{Q^c_{i'} + Q^c_j}{Q^*_{i'} + Q^*_j} \Leftrightarrow
        & \\
        & (Q^c_i + Q^c_j) \cdot (Q^*_{i'} + Q^*_j) < (Q^c_{i'} + Q^c_j) \cdot (Q^*_i + Q^*_j) \Leftrightarrow
        & \\
        & t_1 = Q^c_{i'} \cdot Q^*_i + Q^c_{i'} \cdot Q^*_j + Q^c_j \cdot Q^*_i 
        & \\
        & - Q^c_i \cdot Q^*_{i'} - Q^c_i \cdot Q^*_j - Q^c_j \cdot Q^*_{i'} > 0
    \end{align*}

    And:

    \begin{align*}
        & 1 - \upar^{c}(\investin{c}{i'}{c'}{j}) = \upar^{c'}(\investin{c}{i'}{c'}{j})
        & \\
        & < \upar^{c'}(\investin{c}{i'}{c'}{j'}) = 1 - \upar^{c}(\investin{c}{i'}{c'}{j'}) \Leftrightarrow
        & \\
        & \frac{Q^c_{i'} + Q^c_{j'}}{Q^*_{i'} + Q^*_{j'}} = \upar^{c}(\investin{c}{i'}{c'}{j'})
        & \\
        & < \upar^{c}(\investin{c}{i'}{c'}{j}) = \frac{Q^c_{i'} + Q^c_j}{Q^*_{i'} + Q^*_j} \Leftrightarrow
        & \\
        & (Q^c_{i'} + Q^c_{j'}) \cdot (Q^*_{i'} + Q^*_j) < (Q^c_{i'} + Q^c_j) \cdot (Q^*_{i'} + Q^*_{j'}) \Leftrightarrow
        & \\
        & t_2 = Q^c_{i'} \cdot Q^*_{j'} + Q^c_j \cdot Q^*_{i'} + Q^c_j \cdot Q^*_{j'}
        & \\
        & - Q^c_{i'} \cdot Q^*_j - Q^c_{j'} \cdot Q^*_{i'} - Q^c_{j'} \cdot Q^*_j > 0
    \end{align*}

    And by negation:

    \begin{align*}
        & \frac{Q^c_{i'} + Q^c_{j'}}{Q^*_{i'} + Q^*_{j'}} = \upar^c(\investin{c}{i'}{c'}{j'})
        & \\
        & \le \upar^c(\investin{c}{i}{c'}{j'}) = \frac{Q^c_i + Q^c_{j'}}{Q^*_i + Q^*_{j'}} \Leftrightarrow
        & \\
        & (Q^c_{i'} + Q^c_{j'}) \cdot (Q^*_i + Q^*_{j'}) \le (Q^c_i + Q^c_{j'}) \cdot (Q^*_{i'} + Q^*_{j'}) \Leftrightarrow
        & \\
        & t_3 = Q^c_i \cdot Q^*_{i'} + Q^c_i \cdot Q^*_{j'} + Q^c_{j'} \cdot Q^*_{i'} 
        & \\
        & - Q^c_{i'} \cdot Q^*_i - Q^c_{i'} \cdot Q^*_{j'} - Q^c_{j'} \cdot Q^*_i \ge 0
    \end{align*}

    By \cref{frac:investing_in_higher_issue:prerequisite1} it holds that $Q^c_{i'} > Q^c_i$. Hence, all the terms in the next equation are positive:

    \begin{align*}
        & (Q^c_{i'} + Q^c_{j'}) \cdot t_1 + (Q^c_{i'} - Q^c_i) \cdot t_2 + (Q^c_{i'} + Q^c_j) \cdot t_3
    \end{align*}

    With $(Q^c_{i'} - Q^c_i), t_2$ being strictly positive, the above equation is strictly positive. However, expanding it reveals it's algebraically zero.

    Hence, the negation assumption cannot hold, and it must be that $\upar^c(\investin{c}{i}{c'}{j'}) <
     \upar^c(\investin{c}{i'}{c'}{j'})$.
\end{proof}

\twoCandidatesFocusPureEquilibrium*

\begin{proof}
    We denote $i_k, j_k$ the $k$'th issues $c, c'$ invest in, respectively. $i_1, j_1$ are the issues on which $c, c'$ are ranked the lowest, correspondingly. $i_{k+1}$ is the least higher-ranked response of $c$ to $j_k$ that yields strictly higher utility than $i_k$, and similarly $j_{k+1}$ is the least higher-ranked response of $c'$ to $i_k$ that yields higher utility than $j_k$. This process terminates when no better response exists for both candidates.
    
    We will prove the next claim by induction:

    $\upar^c(\investin{c}{i'}{c'}{j_k}) < \upar^c(\investin{c}{i_k}{c'}{j_k})$ for all $i' \in I$ such that $i_k$ is higher-ranked than $i'$.

    The claim is trivial for $k=0$. Assuming the claim is true for $k$, and that the process is not terminated yet:

    By assumption, $\upar^c(\investin{c}{i'}{c'}{j_k}) < \upar^c(\investin{c}{i_k}{c'}{j_k})$ for all $i' \in I$ such that $Q^*_{i'} < Q^*_{i_k}$.
    By the process' definition $\upar^c(\investin{c}{i'}{c'}{j_k}) < \upar^c(\investin{c}{i_{k+1}}{c'}{j_k})$ for all $i' \in I$ such that $Q^*_{i_k} < Q^*_{i'} < Q^*_{i_{k + 1}}$, and in addition $\upar^c(\investin{c}{i_k}{c'}{j_k}) < \upar^c(\investin{c}{i_{k+1}}{c'}{j_k})$.
    Together, we have $\upar^c(\investin{c}{i'}{c'}{j_k}) < \upar^c(\investin{c}{i_{k+1}}{c'}{j_k})$ for all $i' \in I$ such that $Q^*_{i'} < Q^*_{i_{k+1}}$.

    By definition, $\upar^c(\investin{c}{i_{k+1}}{c'}{j_k})$ > $\upar^c(\investin{c}{i_{k+1}}{c'}{j_{k+1}})$ (since we have a fixed-sum game) and $Q^*_{j_k} < Q^*_{j_{k+1}}$. Now, for every $i': Q^*_{i'} < Q^*_{i_{k+1}}$, we can use \cref{frac:investing_in_higher_issue} and we get that $\upar^c(\investin{c}{i'}{c'}{j_{k+1}}) < \upar^c(\investin{c}{i_{k+1}}{c'}{j_{k+1}})$ for all $i' \in I$ such that $Q^*_{i'} < Q*_{i_{k + 1}}$, and thus complete the induction step.

    The analog proof for the other candidate can be obtained in a similar manner. Together, we can see that when searching for best responds for some candidate, we indeed need to consider only higher-ranked issues than the current investment, necessarily terminating in a (pure, focused) Nash equilibrium.

    For the complexity, by \cref{model:first_simplification} and \cref{model:final_simplification}, computing the simplified version of the problem takes $O(|V|)$ steps.

    The first line of the algorithm takes $O(|I|)$ steps. For the second line, since in every iteration of the loop we advance either $i$ or $j$ and we do not repeat issues, the number of iterations is at most $2 \cdot |I|$. 
    
    Computing the utility of candidates given focus investments is constant (as there are only two candidates), and is done for every issue, hence the first two steps in the loop takes $O(|I|)$ time each. Performing argmax over issues also take $O(|I|)$ time, hence every iteration takes $O(I)$ time, results in total time of $O(|I|^2)$ for the loop.
\end{proof}

\subsection{Proofs from \Cref{sec:presidential}}

\twoCandidatesDominantStrategies*

\begin{proof}
    The difference between the number of votes $c, c'$ get is $f_{\bm{w^{c'}}}(\bm{w^c}) = p(c, \bm{w}) - p(c', \bm{w}) = \bm{Q^c} \cdot \bm{w^c} + \bm{Q^c} \cdot \bm{w^{c'}} - \bm{Q^{c'}} \cdot \bm{w^c} - \bm{Q^c} \cdot \bm{w^{c'}} = b + (\bm{Q^c} - \bm{Q^{c'}}) \cdot \bm{w^c}$

    For $b = (\bm{Q^c} - \bm{Q^{c'}}) \cdot \bm{w^{c'}}$. It can be seen that in order to maximize $f$ with a valid $\bm{w^c}$ it is best to have $w^c_i = W^c$ for $i = \underset{i \in I}{\text{argmax }} Q^c_i - Q^{c'}_i$.

    Now, as the utility of $c$ is monotonic in $f$ (zero for negative value, half for zero, one for one), it is indeed a best response to invest in $\underset{i \in I}{\text{argmax }} Q^c_i - Q^{c'}_i$. As this expression is independent of the investment of $c'$, it must be a dominant strategy.
\end{proof}

\plusTwoCandidatesFracEquilibrium*

\begin{proof}
    Say a Nash equilibrium under $\upar$ is not a Nash equilibrium under $\upresplus$. Then one of the candidates can improve it's utility by investing differently. However, by the definition of $\upresplus$, it means this candidate either improve it's relative fraction of votes $r$, or it's victory indicator $v$. However, improving $r$ implies it's possible to improve the utility under $\upar$, which is a contradiction. Improving $v$ means that the player was losing and is now on tie or better (hence improving it's fraction of votes from less than $0.5$ to $0.5$ or higher), or that the player was on tie and is now winning alone, (hence improving it's share from $0.5$ to higher that this). Either way, we get that one candidate can improve its fraction of votes and hence its utility under $\upar$, which is a contradiction.
\end{proof}

\plusNoBestResponse*

\begin{proof}

Consider the following game.

\begin{table}[h]
\centering
\begin{tabular}{@{\extracolsep{4pt}}rrrr}
\toprule
& $c_1$ & $c_2$ & $c_3$ \\ 
\midrule
$Q^c_1$ & $10$ & $0$ & $11$ \\
$Q^c_2$ & $10$ & $9$ & $9$ \\
$W^c$ & $1$ & $0$ & $0$ \\
\bottomrule
\end{tabular}\\[3mm]
\end{table}

Denoting $\bm{w}(x) = (x, 1-x) : 0 \le x \le 1$ an investment of $x$ in issue $1$ and of $1-x$ in issue $2$, we get that:

\begin{align*}
    & p(c_1, \bm{w}(x)) = 10
    & \\
    & p(c_2, \bm{w}(x)) = 9 \cdot (1-x)
    & \\
    & p(c_3, \bm{w}(x)) = 11 \cdot x + 9 \cdot (1-x)
\end{align*}

It can be seen that $c_1$ wins alone if and only if $x < 0.5$, hence in every best response it must hold. However:

\begin{align*}
    & r(c_1, \bm{w}(x)) = \frac{10}{11 \cdot x + 18 \cdot (1-x)}
\end{align*}

Which is monotonic in $x$ (within the boundary). Hence, the best investment is $\underset{x < 0.5}{\text{max }} x$, which clearly does not exist.

\end{proof}

\maxPureBestResponseExists*

\begin{proof}
    If, in the best response under $\upresind$ (exists by \Cref{ind:best_response_exists}), the responding candidate gets a tie or wins, it is also a best response under $\upresmax$, as the responding candidate is indifferent to his fraction of votes if he doesn't lose.

    Otherwise, the responding candidate cannot win or get a tie, and can only maximize his fraction of votes. In that case the best response under $\upar$ (exists by \Cref{frac:best_response_exists}) is also a best response under $\upresmax$.
\end{proof}

\maxBestResponseDoesntExist*

Proved together with:

\maxNashEquilibriumDoesntExist*

\begin{proof}

Consider the following game:

\begin{table}[h]
\centering
\begin{tabular}{@{\extracolsep{4pt}}rrrrr}
\toprule
& $c_1$ & $c_2$ & $c_3$ & $c_4$ \\ 
\midrule
$Q^c_1$ & $1$ & $1$ & $1 - \epsilon$ & $1$ \\
$Q^c_2$ & $1$ & $1 - \epsilon$ & $1 + \epsilon$ & $0$ \\
$W^c$ & $1$ & $1$ & $0$ & $0$ \\
\bottomrule
\end{tabular}\\[3mm]
\end{table}

We will show that there exists an $\epsilon > 0$ such that the game has no equilibrium.

We first claim that for small enough $\epsilon$, when considering only the fraction of votes, $r$, investing in issue $2$ strictly dominates investing in $1$ for both $c_1, c_2$. This can be seen for $\epsilon = 0$ by \Cref{frac:superior_issues_are_better}. Now, since the fraction of votes is a continuous function, it must be true for positive small enough $\epsilon$ too.

If $c_1$ invests deterministically all its budget in $1$, then the only best response for $c_2$ it to invest all in $1$ and get a tie. However, in that case investing all in $1$ is not a best response for $c_1$ (e.g. it is better to invest $0.9$ in $1$ and $0.1$ in $2$ and win alone), hence it cannot be an equilibrium.  

If $c_1$ invests deterministically all its budget in $2$, then the only best response for $c_2$ it to invest all in $2$. However, in that case investing all in $1$ is a better response for $c_1$, hence it cannot be an equilibrium.  

If $c_2$ invests deterministically all its budget in $1$, then the only best response for $c_1$ is to invest part of its budget in $1$ and part in $2$, in order to be the only winner. However, in that case $c_2$ loses and it's utility is determine only by its fraction of votes, a case in which he would prefer to invest in $2$ instead. Hence, this case too cannot be an equilibrium. 

Similarly, it can be seen that in an equilibrium $c_2$ is not investing part of it's budget in $1$ and part in $2$, because in this case he is losing and would prefer to invest all in $2$.

Denoting $x > 0, 1-x > 1$ the probabilities of $c_2$ to invest all in $1, 2$ correspondingly, the utility of $c_1$ is:

$x \cdot \frac{V}{2} + (1 - x) \cdot V$ when investing all in $1$.

$x \cdot V + (1 - x) \cdot r(c_1, (a, 2 - a))$ when investing $0 < a < 1$ in $1$ and a positive amount in $2$.

$x \cdot \frac{V}{2} + (1 - x) \cdot r(c_1, (0, 2))$ when investing all in $2$, which is strictly dominated by the first option.

Since the only options that are not strictly dominated by others are the first two, and since as was claimed before in an equilibrium it cannot be that $c_1$ invests deterministically in $1$, it must be that $c_1$ has a positive chance to invest $0 < a < 1$ in $1$ and a the rest in $2$, with expected utility of $x \cdot V + (1 - x) \cdot r(c_1, (a, 2 - a))$.

However, from \Cref{frac:superior_issues_are_better}, it can be seen that $r(c, (a, 2 - a))$ is strongly monotonic in $a$, and no maximum limited to $0 < a < 1$ exists. Hence, no equilibrium exists.

It can be seen that for some $x$, a best response does not exist too; i.e., for every $x, a$ close enough to $1$ it must be that $x \cdot V + (1 - x) \cdot r(c_1, (a, 2 - a))$ is greater than $x \cdot \frac{V}{2} + (1 - x) \cdot V$, however there is still no maximum since $a < 1$, hence no best response exists.

\end{proof}

\end{document}